\documentclass[11pt,USenglish]{article}
\usepackage{xcolor} 
\usepackage{cite} 
\usepackage{xspace} 
\usepackage{fullpage}

\usepackage{amsmath,amssymb,amsfonts,amsthm}
\usepackage{graphicx}
\usepackage[pagebackref]{hyperref}
\graphicspath{{./figs/}}
\newcommand{\R}{\mathbb{R}}
\newcommand{\eps}{\varepsilon}
\newcommand{\etal}{et al.\xspace}

\DeclareMathOperator{\DT}{DT}
\DeclareMathOperator{\VD}{VD}
\DeclareMathOperator{\EMST}{EMST}

\newtheorem{theorem}{Theorem}[section]
\newtheorem{lemma}[theorem]{Lemma}

\newtheorem{problem}{Open Problem}

\newcommand{\abbrev}[2]{\expandafter\newcommand\csname #1\endcsname{#2\xspace}}
\newcommand{\cclasss}[2]{\abbrev{#1}{\textsf{#2}}}

\cclasss{LOGSPACE}{LOGSPACE}

\title{Geometric Algorithms with Limited Workspace: A Survey\thanks{BB and WM
supported in part by DFG grant MU 3501/2. MK supported in part by  
MEXT KAKENHI No.~17K12635. WM supported in part
by DFG grant MU 3501/1 and ERC StG 757609.}}

\author{
Bahareh Banyassady\thanks{%
  Institut f\"ur Informatik, Freie Universit\"at Berlin, Germany.
  Email: \texttt{[bahareh,mulzer]@inf.fu-berlin.de}} 
\and
Matias Korman\thanks{%
  Tohoku University, Sendai, Japan.
  Email: \texttt{mati@dais.is.tohoku.ac.jp}}
\and 
Wolfgang Mulzer\footnotemark[2]
}

\begin{document}

\maketitle

\begin{abstract}
In the \emph{limited workspace} model, we consider algorithms whose
input resides in read-only memory and that use only a constant or 
sublinear amount of writable memory to accomplish their task. 
We survey recent results in computational geometry that 
fall into this model and that
strive to achieve the lowest possible running time. In addition to 
discussing the state of the art, we give some illustrative
examples and mention open problems for further research.
\end{abstract}

\section{Introduction}

Space usage has been a concern since the very early days of algorithm 
design. The increased availability of devices with limited memory or 
power supply---such as smartphones, drones, or small sensors---as well as 
the proliferation of new storage media for which write access
is comparatively slow and may have negative effects on the lifetime---such
as flash drives---has led to renewed interest in the subject.
As a result, the design of algorithms for the 
\emph{limited workspace} model
has seen a significant rise in popularity in computational
geometry over the last decade. In this setting, we 
typically have a large amount of data that needs 
to be processed. Although we may access the data in any way and as often 
as we like, write-access to the main storage is limited and/or
slow. Thus, we opt to use only higher level memory for 
intermediate data (e.g., CPU registers).
Since application areas of the devices mentioned above -- sensors,
smartphones, and drones -- often handle a large amount of geographic (i.e.,
geometric) data, the scenario becomes particularly interesting from the
viewpoint of computational geometry.

Motivated by these considerations, there have been numerous recent
works developing algorithms for geometric problems that, in 
addition to the input, use only a small amount of memory. 
Furthermore, there has been research on \emph{time-space trade-offs}, 
where
the goal is to find the fastest algorithm for a given space budget.
In the following, we give a broad overview of these results and of the
current state of the field. We also provide examples to
illustrate the main challenges that these algorithms must face and the
techniques that have been developed to overcome them. 

\section{The Limited Workspace Model}

Designing algorithms that require little working memory is a 
classic and well-known challenge in theoretical computer science.
Over the years, it has been attacked from many different angles.
In \emph{computational complexity theory}, the 
complexity class \LOGSPACE contains all algorithmic problems 
that can be solved with a workspace that has a logarithmic
number of bits~\cite{AroraBa09,Goldreich08}. The 
research on \LOGSPACE has led to several surprising 
insights~\cite{Immerman88,Szelepcsenyi87,Savitch70}, 
perhaps most recently the $st$-connectivity algorithm for 
undirected graphs by Reingold~\cite{Reingold08}, an unexpected 
application of expander graphs. The focus in
computational complexity is mainly on what can be done
\emph{in principle}
in \LOGSPACE. Obtaining \LOGSPACE-algorithms
with a low running time is usually a secondary concern.
\emph{Streaming algorithms} are another classic area where the
amount of writable memory is limited~\cite{Muthukrisnan05}. 
Here, the data items
can be read only once (or a limited number of times), in
an unknown order. In addition, the algorithm may maintain
a sublinear amount of storage to accomplish its task.
There are several results on geometric problems in the
streaming model~\cite{Indyk04,Chan06c}, mostly dealing with 
problems concerning clustering and extent 
measures~\cite{AgarwalHPVa04},
but also with classic questions, such as computing 
convex hulls or low-dimensional linear programming~\cite{ChanCh07}.
The \emph{in-place model} assumes that the input
resides in a memory that can be read and written arbitrarily.
The algorithm may use only a constant number of
additional memory cells.
This means that all complex data 
structures need to be encoded with the help of the
input elements, severely restricting the algorithmic
options at our disposal. There are geometric in-place algorithms
for computing the convex hull or 
the Voronoi diagram of a given planar point 
set~\cite{BronnimannIaKaMoMoTo02,BronnimannChCh04,ChanCh08,ChanCh10}.
In \emph{succinct data structures},
the goal is to minimize the precise number of bits
that are needed to represent the data, getting as close to the
entropy bound as possible~\cite{Navarro16}. At the same time, 
one would like to 
retain the ability to support the desired data structure 
operations efficiently.
In computational geometry, succinct data
structures have been developed for classic problems like range searching,
point location, or nearest neighbor search~\cite{He13}.

The present notion of a \emph{limited workspace algorithm},
which constitutes the main focus of this survey, was 
introduced to the computational geometry community by 
Tetsuo Asano~\cite{Asano08}. 
Initially, the model postulated a workspace that consists
of a constant number of cells~\cite{AsanoMuRoWa11,AsanoMuWa11}.
Over the years, this 
was extended to also allow for time-space trade-offs.
In the following, we describe the most general variant of the model. 

The model is similar to the standard word RAM in which 
the memory is organized as a sequence of \emph{cells}
that can be accessed in constant time via their \emph{addresses}. 
Each memory cell stores a single data word~\cite{Knuth97}. 
In contrast to the standard word RAM, the limited workspace 
model distinguishes two kinds of cells: (i) \emph{read-only cells} 
that store the input; and (ii) \emph{read-write cells} 
that constitute the algorithm's \emph{workspace}. 
A cell of the workspace can store either an integer 
of $O(\log n)$ bits, a pointer to some input cell, or a 
root of some polynomial of bounded degree that depends 
on a fixed number of input variables (for example, the 
intersection point of two lines, each passing through 
two input points). Here, $n$ denotes the input size (measured
in the number of cells).

\begin{figure}
  \centering
    \includegraphics{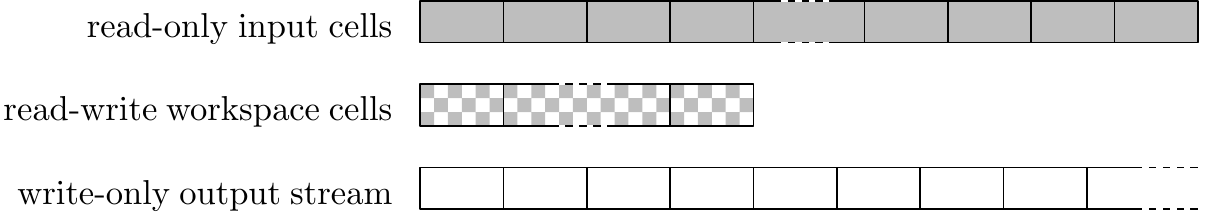}
\caption{The different types of memory 
cells available in the limited workspace model.}
\label{fig:mem-cells}
\end{figure}

Typically, the output is larger than the workspace.
Thus, we assume that the output is written sequentially 
to a dedicated write-only stream. Once a data item
is written to the output stream, it cannot be accessed
again by the algorithm; see Figure~\ref{fig:mem-cells}. It usually 
depends on the algorithmic problem at hand
how exactly the output should be structured.
Our algorithms
are typically deterministic, but there are also results
that use randomization. Randomized algorithms in the
limited workspace model may use an unlimited stream of
random bits. However, these bits cannot be accessed 
arbitrarily: if the algorithm wishes to revisit previous 
random bits, it needs to store them in its workspace.\footnote{Refer
to Goldreich's book~\cite{Goldreich08} for further discussion of 
randomness in the presence of space restrictions.}
As usual, the running time of an algorithm on a given
input is measured as the number of elementary operations
that it performs. The space usage is counted as the number of
cells in the workspace. Note that the input does not contribute to 
this count, but any other memory consumption does (like memory 
implicitly allocated during recursion). 

To bring our model into perspective, we compare it with the related 
models. Unlike the typical viewpoint from computational complexity theory,
our goal is to find the best running time that can be
achieved with a given space budget. In contrast to streaming
algorithms, we may read the input repeatedly and with random access.
Unlike in-place algorithms, our input resides in read-only 
memory, and the workspace can potentially contain arbitrary
data. When analyzing the space usage of an algorithm, we
typically ignore constant factors and lower order terms, whereas 
these play a crucial role in succinct data structures.

Although the objective is to have algorithms that are fast and 
at the same time use little---ideally constant---workspace, it 
is normally not possible to achieve both goals simultaneously.
Thus, the aim is to balance the two. Often, this 
results in a \emph{trade-off}:
as more cells of workspace become
available, the running time decreases.
Now, the precise relationship between the running time and 
the available space becomes the main focus of our attention. 
For many problems,
this dependency is linear (i.e., by doubling the amount of workspace, 
the running time can be halved). However, recent research has
uncovered a wide range of possible trade-offs that often interpolate
smoothly between the best known results for constant and for 
$O(n)$ cells of workspace.

Table~\ref{tab:summary} summarizes recent algorithms for geometric
problems in the limited workspace model. 
In the rest of the survery, we will briefly discuss 
the problems that have been considered and the current state of the 
art.  

\begin{table}[htb]
\centering
\begin{tabular}{llcl}
Problem & Running Time & Space & Source \\ 
\hline
shortest path in a tree & $n$ & $1$ & \cite{AsanoMuWa11} \\ 
all nearest larger neighbors & $n\log_s n$ &$s$ & \cite{AsanoKi13}\\ 
sorting & $n^2/{(s\log n)} + n\log s$ & $s$ & \cite{AsanoElKa13}\\
\hline
convex hull of a point set& $n^2/{(s\log n)} + n\log s$ & $
s$ & \cite{DarwishEl14}  \\ 
triangulation of a point set & $n^2/s + n\log s$ & $s$ &
\cite{AsanoMuRoWa11,AhnBaOhSi17} \\ 
Voronoi diagram/Delaunay triangulation & $n^2\log s/s$ & 
$s$ & \cite{KormanMuvReRoSeSt17,BanyassadyKoMuvReRoSeSt17}\\
Voronoi diagrams of order $1$ to $K$ ($K\leq \sqrt{s}$)& 
$n^2K^6 \text{polylog}(s)/s$ & 
$s$ & \cite{BanyassadyKoMuvReRoSeSt17}\\
Euclidean minimum spanning tree & $n^3 \log s/s^2$ & $s$ & 
\cite{AsanoMuRoWa11,BanyassadyBaMu18} \\
\hline
triangulation of a simple polygon& $n^2/s$ & 
$s$ & 
\cite{OhAh17,AronovKoPrvReRo16} \\
balanced partition of a simple polygon & $n^2/s$ & $s$ & \cite{OhAh17}\\ 
shortest path in a simple polygon& $n^2/s$ & $s$ & 
\cite{AsanoMuWa11,HarPeled16,OhAh17}  \\
triangulation of a monotone polygon& $n\log_s n$ &$s$ & 
\cite{AsanoKi13}   \\ 
\hline
visibility in a simple polygon & $n^2/2^s + n\log^2{n}$ & 
$s$ & \cite{BarbaKoLaSi14} \\
$k$-visibility in a polygonal domain & $n^2/s + n\log s$ & $s$ &
\cite{BahooBaBoDuMu17}  \\
weak visibility in a simple polygon & $n^2$ & $1$ & \cite{Abrahamsen13}\\
minimum link path in a simple polygon& $n^2$ & $1$ & \cite{Abrahamsen13} \\
convex hull of a simple polygon & $n^2\log n/2^s$ & $s^*$  
& \cite{BarbaKoLaSaSi15}  \\
convex hull of a simple polygon & $n^{1+1/\log s}$ & $s^{**}$  
& \cite{BarbaKoLaSaSi15}  \\
common tangents of two disjoint polygons & $n$ & $1$ & 
\cite{Abrahamsen15,AbrahamsenWa16}   \\
\end{tabular}

\caption{A selection of problems and the best known running times
in the limited workspace model. The $O$-notation 
has been omitted in the bounds. If the space usage is given as 
$s$, then $s$ may range from $1$ to $n$. For $s^*$, it may range from 
$1$ to $o(\log n)$, and for $s^{**}$ it ranges from $\log n$ to 
$n$. The running times for $k$-visibility 
and for the convex hull of a simple polygon have been simplified.}
\label{tab:summary}
\end{table}

\section{Point Sets in the Plane}

The most basic problems in computational geometry 
concern point sets in the plane: convex hulls, triangulations,
Voronoi diagrams, Euclidean minimum spanning trees, and
related structures have captured the imagination of computational
geometers for decades~\cite{deBergChvKrOv08}. Naturally, these structures 
have
also been an early focus in the study of the limited workspace model.

\subsection{Convex Hulls} 
Asano~\etal~\cite{AsanoMuRoWa11} observed that the edges of
the convex hull for a set of $n$ points in the plane can
be found in $O(n^2)$ time when $O(1)$ cells of workspace
are available, through a straightforward application of 
Jarvis' classic gift-wrapping algorithm~\cite{Jarvis73}.
Darwish and Elmasry~\cite{DarwishEl14} extended this result to 
an asymptotically optimal time-space
trade-off. For this, they developed a 
space-efficient heap data structure. 
The optimality of the trade-off is implied
by a lower bound for the sorting problem due to Beame~\cite{Beame91}. 
The algorithm by Darwish and Elmasry needs
$O(n^2/(s \log n) +n\log s)$ time in order to output the edges 
of the convex hull, provided that
$s$ cells of workspace are available. The underlying heap data structure 
is very versatile, and it can also
be used to obtain time-space trade-offs for the sorting
problem and for computing a triangulation of a planar
point set~\cite{AsanoElKa13,KormanMuvReRoSeSt17}.

\subsection{Delaunay Triangulations and Voronoi Diagrams}

\paragraph*{Constant Workspace.}
For computing the Delaunay triangulation $\DT(S)$  and the Voronoi diagram 
$\VD(S)$ of a given set $S= \{p_1, \dots, p_n\}$ of $n \geq 3$ sites 
in the plane, 
Asano~\etal~\cite{AsanoMuRoWa11} presented an $O(n^2)$-time 
algorithm that uses $O(1)$ cells of workspace under a general 
position assumption (i.e., no 3 sites of $S$ lie on
a common line and no 4 sites of $S$ lie on a common
circle).  The \emph{Voronoi diagram}
for $S$, $\VD(S)$, is obtained by classifying the points
in the plane according to their nearest neighbor in
$S$. For each site $p \in S$, the open set
of points in $\R^2$ with $p$ as their unique nearest
site in $S$ is called the \emph{Voronoi cell} of $p$ and is 
denoted by $C(p)$.
The \emph{Voronoi edge} for two sites $p, q \in S$ consists of all 
points in the plane with $p$ and $q$ as their only two nearest sites.
It is a (possibly empty) subset of the \emph{bisector}
$B(p,q)$ of $p$ and $q$, the line that has all points 
with the same distance from $p$ and from $q$.
Finally, \emph{Voronoi vertices} are the points in the plane 
that have exactly 
three nearest sites in $S$. It is well known that $\VD(S)$ has 
$O(n)$ vertices, $O(n)$ edges, and $n$ cells~\cite{deBergChvKrOv08}.

First, we explain how 
Asano~\etal~\cite{AsanoMuRoWa11} find a single edge of a given 
cell $C(p)$ of $\VD(S)$.
Then, we repeatedly use this procedure to find all the edges 
of $\VD(S)$, using $O(1)$ cells of workspace. 

\begin{lemma}\label{lem:singleVoronoiEdge}
Given a site $p \in S$ and a ray $\gamma$ from $p$ that 
intersects $\partial C(p)$, we can 
report an edge $e$ of $C(p)$ 
whose relative closure intersects $\gamma$ in $O(n)$ time 
using $O(1)$ cells of workspace.
\end{lemma}

\begin{proof}
Among all bisectors $B(p, q)$, for $q \in S \setminus \{p\}$,
we find a bisector $B^* = B(p, q^*)$ that intersects $\gamma$ closest to 
$p$.  This can be done by scanning the sites of 
$S$ while maintaining a closest bisector.
The desired Voronoi edge $e$ is a subset of $B^*$.
To find $e$, we scan $S$ again.
For each $q \in S \setminus \{p, q^*\}$, we compute 
the intersection between $B(p, q)$ and $B^*$. 
Each such intersection determines a
piece of $B^*$ not in $\VD(S)$, namely
the part of $B^*$ that is closer to $q$ than to $p$. 
The portion of $B^*$ that remains after the scan 
is exactly $e$.
Since the current piece of $B^*$ in each step is connected,
we must maintain at most two endpoints to represent the current
candidate for $e$. The claim follows.
\end{proof}

\begin{theorem}\label{thm:cwsVoronoi}
Suppose we are given $n$ planar sites 
$S=\{p_1, \dots, p_n\}$ in general position.
We can find all the edges of 
$\VD(S)$ in $O(n^2)$ time 
using $O(1)$ cells of workspace.
\end{theorem}

\begin{proof}
We process each site $p \in S$ to find the 
edges on $\partial C(p)$. 
We choose $\gamma$ as the ray from $p$ 
through an arbitrary site of $S \setminus \{p\}$. 
Then, $\gamma$ intersects 
$\partial C(p)$. Using 
Lemma~\ref{lem:singleVoronoiEdge}, we find an edge $e$ 
of $C(p)$ that intersects $\gamma$.
We let $\gamma'$ be the ray from $p$ through the left endpoint 
of $e$ (if it exists), and we apply Lemma~\ref{lem:singleVoronoiEdge} 
again
to find the left adjacent edge $e'$ of $e$ in $C(p)$.
We repeat this to find further
edges of $C(p)$, in counterclockwise order,
until we return to $e$ or until we find an unbounded edge of $C(p)$. 
In the latter case, we start again from the right 
endpoint of $e$ (if it exists), and we find the remaining
edges of $C(p)$, in clockwise order. 

Since each Voronoi edge is incident to two Voronoi cells, 
we will encounter each edge twice. 
To avoid repetitions in the output, whenever we find an edge $e$ 
of $C(p_i)$ with $e \subseteq B(p_i, p_j)$, we 
report $e$ if and only if $i<j$. 
Since $\VD(S)$ has $O(n)$ edges, and since
reporting one edge takes $O(n)$ time and $O(1)$ cells 
of workspace, the result follows.
\end{proof}

\paragraph*{Time-Space Trade-Off.}
Korman~\etal~\cite{KormanMuvReRoSeSt17}
gave a randomized time-space trade-off for computing 
$\VD(S)$ that runs in expected time 
$O((n^2/s)\log s+n \log s \log^*s)$, provided that
$s$ cells of workspace may be used. 
The algorithm is based on  a space-efficient implementation of the
Clarkson-Shor random sampling technique~\cite{ClarksonSh89}
that makes it possible to divide the problem into $O(s)$ 
subproblems with $O(n/s)$ sites each. 
All subproblems can then be handled simultaneously
with the constant workspace method of Asano~\etal~\cite{AsanoMuRoWa11}, 
resulting in the desired running time. 

More recently, Banyassady~\etal~\cite{BanyassadyKoMuvReRoSeSt17} 
developed a better trade-off  that is also applicable
to Voronoi diagrams of a higher order.
Their algorithm is deterministic, and it computes the
traditional Voronoi diagram as well as the farthest-site Voronoi diagram
of $S$ in $O(n^2\log s/s)$ time, using $s$ cells of workspace,
for some $s \in \{1, \dots, n\}$.
The main idea is to obtain $\VD(S)$ by processing
$S$ in \emph{batches} of $s$ sites each, using a special
procedure to handle sites whose Voronoi cells have a large number of
edges. 

We now describe this algorithm in more detail. 
We assume that our algorithm has
$O(s)$ cells of workspace, for the given parameter $s$.
This does not strengthen the model, and it simplifies
the analysis of the algorithms, as we can assume that we
can manipulate structures that require $O(s)$ cells of space
(such as a Voronoi diagram of $O(s)$ sites), without a detailed
analysis of the associated constants.

\begin{lemma}\label{lem:batchVoronoiEdge}
Suppose we are given a set $V = \{p_1, \dots, p_s\}$
of $s$ sites in $S$, together with, for each $i \in \{1, \dots, s\}$, 
a ray $\gamma_i$  from $p_i$ that intersects
$\partial C(p_i)$. We can report, for each $i$, 
an edge $e_i$ of $C(p_i)$  that intersects $\gamma_i$, 
in total time $O(n\log s)$ using $O(s)$ 
cells of workspace.
\end{lemma}

\begin{proof}
The algorithm scans the sites twice.
In the first scan, we find, for $i = 1, \dots, s$, the 
bisector $B^*_i$ that contains $e_i$. In the second scan, we 
determine for each $B^*_i$, $i = 1, \dots, s$, 
the exact portion of $B^*_i$ that determines $e_i$.

We start by partitioning $S$ into $n/s$ disjoint \emph{batches} 
$T_1, T_2, \dots, T_{n/s}$ of $s$ consecutive sites in the input.
Then, we compute $\VD(V \cup T_1)$ in $O(s\log s)$ time, 
using $O(s)$ cells of workspace.  For each $i$,
we find the edge $e_i'$ of $\VD(V \cup T_1)$ that intersects 
$\gamma_i$ closest to $p_i$, and we 
store the bisector $B_i'$ that contains $e_i'$. 
This takes time $O(|V \cup T_1|) = O(s)$ by a simple
traversal of the diagram.
Now, for $j = 2, \dots, n/s$,
we repeat this procedure with $\VD(V \cup T_j)$, and we update
$B_i'$ if the new diagram gives an edge
that intersects $\gamma_i$ closer to $p_i$ than $B_i'$. 
After all batches $T_1, \dots, T_{n/s}$ have been scanned, 
the current bisector $B_i'$ is the desired final bisector $B_i^*$, for 
$i = 1, \dots, s$.

In the second scan, to find the portion of $B^*_i$ that constitutes 
$e_i$, for $i = 1, \dots, s$, we again consider the batches 
$T_1, \dots, T_{n/s}$ of size $s$.
As before, for $j = 1, \dots, n/s$, we compute 
$\VD(V \cup T_j)$, and 
for $i = 1, \dots, s$, we find the portion 
of $B_i^*$ inside the cell $C_{ij}$ of $p_i$ in $\VD(V \cup T_j)$. 
We update the endpoints of $e_i$ to the 
intersection of the current $e_i$ and the cell $C_{ij}$.
After processing $T_j$, 
there is no site in $V \cup \bigcup_{k=1}^{j} T_k$  that is 
closer to $e_i$ than $p_i$.
Thus, after the second scan, $e_i$ 
is the edge of $C(p_i)$ that intersects $\gamma_i$. 

In total, we construct $O(n/s)$ Voronoi diagrams, each with at most
$2s$ sites, in $O(s \log s)$ time each. The remaining work in each step
of each scan is $O(s)$.
Thus, the total running time is $O(n\log s)$. At each time, we 
have $O(s)$ sites in workspace and a constant amount of information 
per site, including their Voronoi diagram. Thus, 
the space bound is not exceeded.
\end{proof}

The global time-space trade-off algorithm has two phases. 
In the first one, we scan $S$ sequentially. During this scan, 
we keep a set $V$ of $s$ sites from $S$ in the workspace whose
Voronoi cells we intend to find.
In the beginning, $V$ consists of the first $s$ sites from $S$,
and we apply Lemma~\ref{lem:batchVoronoiEdge} to 
compute one edge of $\VD(S)$ on each cell $C(p_i)$, for 
$p_i \in V$.
The starting ray $\gamma_i$ for each $p_i \in V$
is constructed in the same way as 
in Theorem~\ref{thm:cwsVoronoi}. After that, we 
update these rays to find the next 
edge on each $C(p_i)$, for $p_i \in V$, as in 
Theorem~\ref{thm:cwsVoronoi}.
Now, however, for each $p_i \in V$, we store the original ray 
$\gamma_i'$ in addition to the current ray $\gamma_i$. 
These two rays help us determine
which edges on $C(p_i)$ have been reported already.
Whenever all edges for a site $p_i \in V$ have been found, we replace 
$p_i$ with the next relevant site from $S$, and we say that $p_i$ has 
been \emph{processed}; see Figure~\ref{fig:VD}.

\begin{figure}
  \centering
    \includegraphics{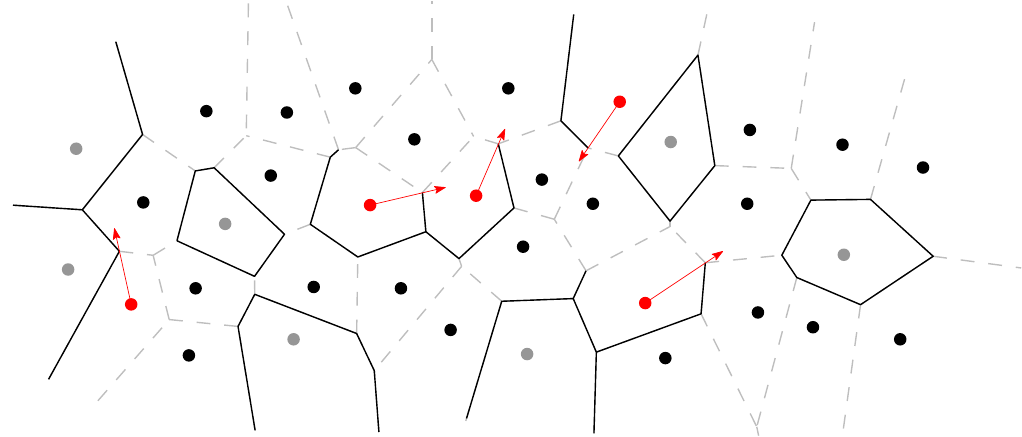}
\caption{Illustration of the algorithm of 
Banyassady~\etal~\cite{BanyassadyKoMuvReRoSeSt17} after nine iterations of 
Lemma~\ref{lem:batchVoronoiEdge} for a set $S$ of $n = 35$ sites and 
workspace of $O(\log n)$ cells. The black segments are the 
edges of $\VD(S)$ that have already been found. The gray and the red 
sites represent, respectively, the sites which have been fully processed
and those which are currently in the workspace.}
\label{fig:VD}
\end{figure}

Since (i) the Voronoi diagram of $S$
has $O(n)$ edges;  (ii) in each iteration, we find $s$ edges;
and (iii) each edge is encountered at most twice, it follows that after 
$O(n/s)$ iterations of this procedure, fewer than $s$ sites remain in 
$V$. All other sites of $S$ have been processed

At this point, phase 1 of the algorithm ends 
(and the second one starts). During the execution of the first 
phase, we output only some of the 
Voronoi edges, according to the following rule: suppose we discover 
the edge $e$ while scanning the site $p_i$, and let $p_j$ be the
other site with $e$ on its cell boundary.
We output $e$ only if either (i) $p_j$ is a \emph{fresh} site, i.e., 
$p_j$ has not been processed yet and is not currently in $V$
(this can be tested using the index of the last site inserted into $V$); 
or (ii) $p_j$ is in $V$ and $e$ 
has not been reported as an edge of $C(p_j)$ (this can be tested in 
$O(\log s)$ time with a binary search tree that contains all
elements in $V$ and using 
the rays $\gamma_j$ and $\gamma'_j$ ). In this way, each Voronoi edge
is reported at most once.
 
Let $V_R$ be the set of sites that have not been processed 
when the first phase of the algorithm ends.
We cannot use Lemma~\ref{lem:batchVoronoiEdge} any longer, 
since it needs two passes over the input to find 
a single new edge for each site in $V_R$, and the sites in $V_R$ may
have too many associated Voronoi edges. 
Each remaining Voronoi edge must be on the cell boundaries of 
two sites in $V_R$. 
Thus, in the second phase, we compute 
$\VD(V_R)$ in $O(s \log s)$ time.
Let $E_{R}$ denote the set of its edges. Some edges of $E_R$ 
also occur in $\VD(S)$ (possibly in truncated form).
To identify these edges,
we proceed similarly to the second scan of 
Lemma~\ref{lem:batchVoronoiEdge}: in each step, we compute the 
Voronoi diagram $\mathcal{V}$ of $V_R$ and a batch of $s$ sites from 
$S$. For each edge $e$ of $E_R$, we check whether $e$ occurs
in $\mathcal{V}$ (possibly in truncated form). If not, we set
$e$ to be empty; otherwise, we update 
the endpoints of $e$ according to the truncated version.
After all edges in $E_R$ have been checked, 
we continue with the next batch of $s$ sites from $S$. After 
processing all batches, the remaining non-empty edges 
in $E_R$ are precisely the edges of $\VD(S)$
that are incident to two cells in $\VD(V_R)$. 
Some of these edges may have
been reported in the first phase. We can identify each such edge
in $O(1)$ time by using the starting ray and the current ray  
of the two incident cells from the first phase. These rays are
still available in the workspace, because $V_R$ consists of those
sites that were left over at the end of the first phase.
We output the remaining edges in $E_R$.

\begin{theorem}\label{thm:tradeoffVoronoi}
Let $S = \{p_1, \dots, p_n\}$ be a 
planar $n$-point set in general position, and $s \in \{1, \dots n\}$.
We can report all  edges of $\VD(S)$ using 
$O((n^2/s) \log s)$ time and $O(s)$ cells 
of workspace. 
\end{theorem}

\begin{proof}
Lemma~\ref{lem:batchVoronoiEdge} guarantees that the edges reported 
in the first phase are in $\VD(S)$. Also, conditions (i) and 
(ii) ensure that no edge is reported twice. Furthermore, in the 
second phase, no edge will be reported for a second time.
Since $V_R \subset S$, an 
edge $e \in \VD(S)$ is incident to the cell of two sites in $V_R$, 
if and only if the same edge 
(possibly in extended form) occurs in $\VD(V_R)$. 
Furthermore, for each edge $e$ of $\VD(V_R)$, 
we consider all sites of $S$, and we remove only the 
portions of $e$ that cannot be present in $\VD(S)$.

Regarding the running time, the first phase
requires $O(n/s)$ 
invocations of Lemma~\ref{lem:batchVoronoiEdge} and $O(n)$ tests whether
a Voronoi edge should be output. This takes $O((n^2/s)\log s + n\log s)$ 
time.  The second phase does a single scan over 
$S$, and it computes a Voronoi diagram for each batch of 
$s$ sites, which takes $O(n\log s)$ 
time in total. Thus, the total running time 
is $O((n^2/s)\log s)$.

At each point, we store 
only $s$ sites in $V$ (along with a constant amount of  
information attached to each site),
the batch of $s$ sites being processed and the 
associated Voronoi diagram. 
All of this requires $O(s)$ cells of workspace, as 
claimed. 
\end{proof}

Banyassady~\etal~\cite{BanyassadyKoMuvReRoSeSt17}~also showed 
that these techniques 
work for Voronoi
diagrams of higher order~\cite{AurenhammerKlLe13,Lee82}. 
More precisely, they obtained the first time-space trade-off for computing 
the family of all Voronoi diagrams for $S$ up to a given 
order $K$ in $O\big(n^2K^5(\log s+ K\log K)/s\big)$ time using 
$s$ cells of workspace.\footnote{The 
algorithm assumes
that $K = O(\sqrt{s})$. This assumption is due to the fact that
we need $\Theta(K)$ cells of workspace to represent a feature of
a Voronoi diagram of order $K$, so that our algorithm can handle
up to $\sqrt{s}$ features of a diagram of order $\sqrt{s}$ with $s$
cells of workspace.} 
The algorithm is based on
the \emph{simulated parallelization} technique, i.e, there is 
one instance of the
algorithm for each order $k$, for $k = 1, \dots, K$, 
where the
instance for order $k$ outputs the features for the order-$k$
Voronoi diagram and produces the input needed by the instance
for order $k+1$. The computational steps of
the individual instances and the memory usage are coordinated in such a
way as to make efficient use of the available workspace while
avoiding an exponential blow-up in the running time
that could be caused by a naive application of the technique.
\begin{problem}
What is the best trade-off for computing higher-order Voronoi diagrams?
Are there non-trivial lower bounds? Can we quickly compute a diagram
of a given order without computing the diagrams of lower order?
How can we compute Voronoi diagrams of order larger than $\sqrt{s}$
when $s$ words of workspace are available?
\end{problem}

If the goal is to compute any triangulation (rather than 
a Delaunay triangulation), Ahn~\etal~\cite{AhnBaOhSi17} gave 
a divide and conquer approach that is faster. Their method runs 
in $O(n^2/s + n \log s)$ time and uses $s$ cells of workspace. 

\subsection{Euclidean Minimum Spanning Trees}

\paragraph*{Constant Workspace.}
Let $S$ be a set of $n$ sites in the plane.
The \emph{Euclidean minimum spanning tree} of $S$, 
$\EMST(S)$, is the minimum spanning tree of the complete
graph $G_S$ on $S$, where the edges are weighted with the 
Euclidean distance of their endpoints. 
The task of computing $\EMST(S)$
was among the first problems to be
considered in the limited workspace model. 
Asano~\etal~\cite{AsanoMuRoWa11} provided an algorithm that reports 
the edges of $\EMST(S)$
by increasing length using $O(n^3)$ time and $O(1)$ cells 
of workspace. This is still the fastest algorithm for the problem when
constant workspace is available. 

We now describe the algorithm in more detail.
We assume that 
$S$ is in general position, i.e., no sites of $S$ lie on a common line,
no four points in $S$ lie on a common circle, and the edge lengths in 
$G_S$ are pairwise distinct. Then, $\EMST(S)$
is unique. Given $S$, we can compute $\EMST(S)$ in $O(n \log n)$ 
time if $O(n)$ cells of workspace are 
available~\cite{deBergChvKrOv08}.
It is well known that $\EMST(S)$ is a subgraph of
the \emph{Delaunay triangulation} of $S$, 
$\DT(S)$~\cite{deBergChvKrOv08}.  
Thus, it suffices to consider only the edges of 
$\DT(S)$ instead of the complete graph $G_S$ on $S$.

We apply Theorem~\ref{thm:cwsVoronoi},
to construct $\EMST(S)$ using the edges of 
$\DT(S)$~\cite{AsanoMuRoWa11}.  
Since $\DT(S)$ and $\VD(S)$ are graph-theoretic
duals, we can output all edges of $\DT(S)$ in $O(n^2)$ time 
with $O(1)$ cells
of workspace by adapting the algorithm in the proof of 
Theorem~\ref{thm:cwsVoronoi}. Even more, we can
implement a subroutine \texttt{clockwiseNextDelaunayEdge$(p, pq)$} 
that receives the Delaunay edge $pq$ of a site $p$ and finds the next 
clockwise Delaunay edge incident to $p$ in $O(n)$ time using $O(1)$ 
cells of workspace. 

By the \emph{bottleneck shortest
path} property of minimum spanning trees, 
a Delaunay edge $e = pq$ is not in $\EMST(S)$ if and only if 
$\DT(S)$ has a path between $p$ and $q$ 
consisting only of edges with length less than 
$|pq|$~\cite{Eppstein00}.
Let $\DT_{<e}$ be the subgraph of $\DT(S)$ with the 
edges of length less than $|e|$. The subgraph 
$\DT_{\leq e}$ is defined analogously, having all edges of $\DT(S)$ 
with length at most $|e|$.

\begin{lemma}\label{lem:checkEMSTedge}
Let $S$ be a planar point set and let $e = pq$ be an edge of $\DT(S)$.
Given $e$, we can determine whether $e$ appears
in $\EMST(S)$ in $O(n^2)$ time using $O(1)$ cells of workspace.
\end{lemma}

\begin{proof}
As explained above, we must determine
whether there is a path between $p$ and $q$ in $\DT_{<e}$.
Since $\DT_{\leq e}$ is a plane graph, if $p$ to $q$ are connected 
in $\DT_{<e}$, then there is a path 
that forms a face in $\DT_{\leq e}$ with $e$ on its boundary.  
Thus, we can check for the existence of such a path
by walking from $p$ along the boundary of the face of $\DT_{< e}$ 
that is intersected by $e$
and checking whether we can encounter $q$. 

In the first step, we find the clockwise next edge $pr$ from $e$ that is 
incident to $p$ and has length at most $|e|$. This can be done
by repeated invocations of 
\texttt{clockwiseNextDelaunayEdge}. 
In the next step, we find the clockwise next edge from $rp$ that is 
incident to $r$ and that has length at most $|e|$,
again by repeated calls to \texttt{clockwiseNextDelaunayEdge}.
The walk continues, for at most $2n$ steps, until we encounter 
the edge $e$ again.
This can happen in two ways: (i) we see $e$ while enumerating the
incident edges of $q$. Then, $p$ and $q$
are in the same component of $\DT_{<e}$, and hence $e$ does not belong to
$\EMST(S)$; or (ii) we see $e$ while again enumerating incident edges of 
$p$. In this case, it follows that $e$ connects two connected components
of $\DT_{<e}$, and $e$ belongs to $\EMST(S)$.\footnote{Note that
the walk may return to $p$ first even though $p$ and $q$ lie in the
same component of $\DT_{<e}$. Thus, it is crucial that we continue the
walk until $e$ is encountered for a second time from $p$.} 
See Figure~\ref{fig:EMST-Edge} for an illustration.

\begin{figure}
  \centering
    \includegraphics{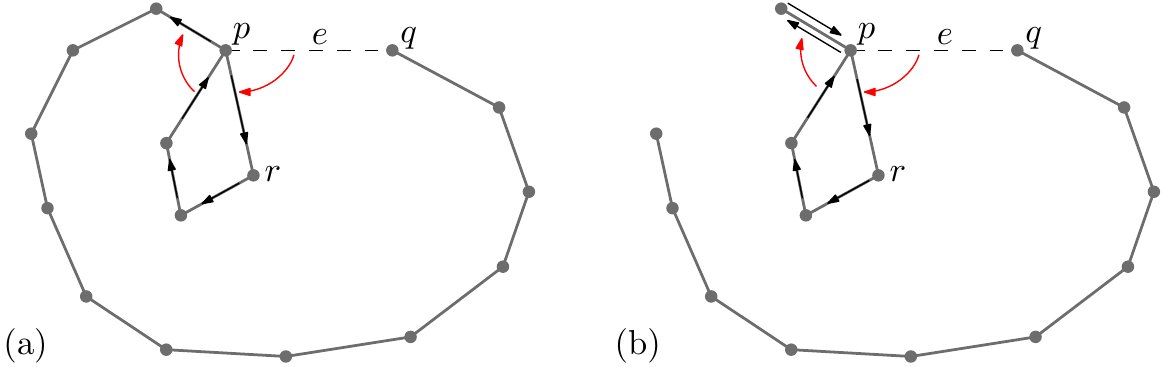}
\caption{Illustration of the algorithm for determining if the edge 
$e=pq$ is part of the EMST.
The walk starts from the edge $pr$. After 5 steps, it returns
to $p$, but it does not see $e$. Thus the walk continues until it 
encounters $e$ as an edge incident to $q$. Thus, 
$e \not \in \EMST(S)$. (b) Slightly modified instance in which $p$
and $q$ are in different components. The walk starts from $pr$,
and after 7 steps it encounters $e$ as an edge of $p$. 
Thus, $e\in \EMST(S)$.}
\label{fig:EMST-Edge}
\end{figure}

This procedure generates a subset of the edges in $\DT(S)$, and
each edge is generated at most twice,
by a call to
\texttt{clockwiseNextDelaunayEdge} in $O(n)$ time. 
Thus, the total running time is $O(n^2)$. The space bound 
is immediate.
\end{proof}

\begin{theorem}\label{thm:cwsEMST}
Given a set $S$ of $n$ sites in the plane, we can 
output all edges of $\EMST(S)$ in $O(n^3)$ 
time using $O(1)$ cells of workspace.
\end{theorem}

\begin{proof}
The algorithm computes the edges of $\DT(S)$ using the adaptation 
of the algorithm in the proof of Theorem~\ref{thm:cwsVoronoi}. 
Every time we detect a new Delaunay edge $e$, we pause the 
computation of $\DT(S)$ and, using Lemma~\ref{lem:checkEMSTedge}, we 
determine if $e$ is in $\EMST(S)$. If so, we output $e$ (do nothing 
otherwise). In either case we resume with the computation of 
$\DT(S)$ (pausing and resuming can be done because this subroutine 
only uses $O(s)$ space). Since
$\DT(S)$ has $O(n)$ edges, and since it takes
$O(n^2)$ time to decide membership in 
$\EMST(S)$, the total time to determine the edges of
$\EMST(S)$ is $O(n^3)$. Furthermore, by 
Theorem~\ref{thm:cwsVoronoi}, the overhead for computing
$\DT(S)$ is $O(n^2)$, which is negligible compared to the remainder of the
algorithm. The space requirement is immediate from 
Theorem~\ref{thm:cwsVoronoi} and 
Lemma~\ref{lem:checkEMSTedge}.  
\end{proof}

\begin{problem}
Can $\EMST(S)$ be found in $o(n^3)$ time and constant workspace?
\end{problem}

\paragraph*{Time-Space Trade-Offs.}
The running time of the algorithm is dominated by the $O(n)$ 
calls to Lemma~\ref{lem:checkEMSTedge}. Thus, the time-space trade-off
from Theorem~\ref{thm:tradeoffVoronoi} does 
not immediately extend for computing $\EMST(S)$. 
Recently, Banyassady~\etal~\cite{BanyassadyBaMu18} 
revisited the problem and provided a time-space trade-off, 
building on the ideas in Theorem~\ref{thm:cwsEMST}.
Their algorithm computes $\EMST(S)$
in $O(n^3 \log s/ s^2)$ time, provided that $s$ cells of workspace are
available. 
It uses the workspace in two different ways: akin to 
Lemma~\ref{lem:batchVoronoiEdge}, we check $s$ edges in parallel 
for membership in $\EMST(S)$. 
Further, Banyassady~\etal~~\cite{BanyassadyBaMu18} introduce
\emph{$s$-nets}, a compact 
representation of planar graphs.
Using $s$-nets, one can speed up
Kruskal's MST algorithm on $S$ by making better use of the  
limited workspace.\footnote{Although the spirit of the algorithm is 
the same, in~\cite{BanyassadyBaMu18}, the walks are 
performed in the \emph{relative neighborhood graph} of $S$ instead 
of $\DT(S)$. This is critical, 
since $\DT(S)$ is not of bounded degree.}
The $s$-net structure seems to be of independent interest, 
as it provides a compact way to represent planar 
graphs that could be exploited by other algorithms that deal
with such graphs.

\section{Triangulations, Partitions, and 
Shortest Paths in Polygons}

Let $P$ be a simple planar polygon with $n$ vertices. 
When dealing with $P$, it is often useful to first compute a
\emph{triangulation} for $P$. We remind the reader of the terminology: 
a \emph{diagonal} is 
a line segment between two vertices of $P$ that 
goes through the interior of $P$.  A triangulation 
is a maximal set of diagonals that do not 
cross each other.
Any triangulation of $P$ contains exactly $n-3$ diagonals,
and---famously and perhaps notoriously---a 
triangulation of $P$ can be found in linear time when
linear workspace is available~\cite{Chazelle91}.

In order to perform divide-and-conquer algorithms on $P$, we often would
like to have a \emph{balanced partition} of $P$:
a \emph{chord} is a line segment with endpoints 
on the boundary of $P$ that goes through the interior of $P$. 
Given $s \in \{1, \dots, n\}$, a balanced partition of $P$ is a 
set of $s$ mutually non-crossing \emph{chords} that 
partition $P$ into subpolygons with 
$O(n/s)$ vertices each. Ideally, the 
chords should be diagonals,  
but other choices, such as vertical segments, are 
acceptable as well. With linear workspace at our disposal, 
for any given $s$, a balanced 
partition of $P$ can be obtained in linear time by first 
triangulating $P$ and then traversing the dual graph of 
the triangulation while greedily selecting appropriate 
diagonals to serve as the chords of the partition.

A third classic problem on $P$ is
shortest path computation: given two points $u$ and $v$ inside $P$, 
find the unique (geodesic) shortest path from $u$ to $v$ that stays in 
$P$. 
The standard algorithm for this problem triangulates $P$ and 
then walks in the dual graph of the triangulation from the 
triangle containing $u$ to the triangle containing $v$.
During this walk, the algorithm maintains a simple
\emph{funnel} data structure and outputs the segments of the
shortest path. This requires
linear time and linear space (see, for example, Mitchell's 
survey~\cite{m-spn-17}).

Thus, in traditional computational geometry, it
appears that the most fundamental of the
three problems is polygon triangulation. Once a 
subroutine for triangulating $P$ is at hand, the 
other two problems can be solved
rather easily, with additional linear overhead. However, the
study of the limited workspace model has revealed a more intricate 
web of relationships between these problems 
that we will now explain.

\subsection{First Results in the Limited Workspace Model}
Asano~\etal~\cite{AsanoMuWa11} showed a way of navigating a 
\emph{trapezoidation} of $P$ with $O(1)$ cells of workspace,
so that the adjacent trapezoids for any given trapezoid can be
found in $O(n)$ time (assuming that the vertices of $P$ are in general
position). They also gave a simple algorithm to output the path 
between two vertices of a tree that needs linear time and 
$O(1)$ cells of workspace~\cite{AsanoMuWa11}. These two results together 
lead to a constant workspace algorithm
for the shortest path problem in $P$ that requires $O(n^2)$ time
in theory and turns out to be quite efficient in 
practice~\cite{CleveMu17}. 
Furthermore, Asano~\etal provide two alternative
constant workspace
algorithms for shortest paths in polygons that are based on 
constrained Delaunay triangulations~\cite{AsanoMuWa11} and 
ray shooting~\cite{AsanoMuRoWa11}, respectively. The former
algorithm requires $O(n^3)$ time, the latter algorithm requires
$O(n^2)$ time. 
An experimental evaluation of these algorithms
was conducted by Cleve and Mulzer~\cite{CleveMu17}. It  showed that
the theoretical guarantees for the shortest path algorithms 
could also be observed in practice. 
\begin{problem}
How fast can we find shortest paths in polygonal domains
in the limited workspace model?
\end{problem}

\subsection{Space-efficient Reductions}
The initial study of Asano~\etal~\cite{AsanoMuWa11} 
was quickly followed by another
work of Asano~\etal~\cite{AsanoBuBuKoMuRoSc13} that showed how 
to enumerate all triangles in a triangulation of $P$
using $O(n^2)$ time and $O(1)$ cells of workspace. They also proved 
the following:
if one can compute a triangulation of $P$ in time $T(n, s)$ when $s$
cells of workspace are available, then one can use a bottom-up approach to 
find $O(s)$ diagonals that constitute a balanced partition of $P$ 
into subpolygons with $O(n/s)$ vertices, using  $T(n, s)+ O(n^2/s)$ 
time and $s$ cells of workspace. As they explain, such a balanced 
partition also leads to
a time-space trade-off for the shortest path problem in $P$: given the
diagonals of the partition, one can navigate through the 
subpolygons and compute a 
shortest path between any two vertices in $P$  in additional 
$O(n^2/s)$ time, using $s$ cells of workspace.
Thus, we have a reduction from shortest paths in polygons to 
polygon triangulation
that runs in $T(n,s) + O(n^2/s)$ time, using $s$ cells of 
workspace. At the time of Asano~\etal's work~\cite{AsanoBuBuKoMuRoSc13},
the best general bound for $T(n,s)$ was $T(n,s) = O(n^2)$, so that
their result at first seemed only useful in a regime where
$P$ can be preprocessed and the balanced partition can be stored
in the workspace for later use.

Aronov~\etal~\cite{AronovKoPrvReRo16} showed a connection between
shortest paths and polygon triangulation that goes
in the other direction.
Given a shortest path algorithm as a black-box, they use it 
to partition $P$ recursively into smaller pieces
that fit into the workspace. Then, each piece 
can be triangulated in the workspace using linear time and 
space~\cite{Chazelle91}. Even though it may happen that the recursion
runs out of space---in which case the algorithm needs to fall back
to a brute force triangulation method---it turns out that if 
there are at least
$s = \Omega(\log n)$ cells of workspace at our disposal, the recursion
always succeeds. Then,
the resulting running time $R(n)$ is of the 
form $R(n) = T(n, s) + R(cn)$, where $c < 1$ is a constant and 
$T(n, s)$ is the time needed to compute the shortest path between 
two points in a simple polygon with $n$ vertices, provided that $s$ cells
of workspace are available. 

Thus, by combining the results of Asano~\etal~\cite{AsanoBuBuKoMuRoSc13} 
and Aronov~\etal~\cite{AronovKoPrvReRo16}, one can see that all three 
problems are equivalent in the limited workspace model: 
a triangulation can be used to partition 
$P$ into balanced pieces~\cite{AsanoBuBuKoMuRoSc13}. Once we 
have the partition, 
we can compute the shortest path between any two 
points in $P$~\cite{AsanoBuBuKoMuRoSc13}. Finally, a 
shortest path subroutine 
allows us to find a triangulation of $P$ in essentially the same 
time and 
space~\cite{AronovKoPrvReRo16}. 
Thus, given a fast 
algorithm for any of the three problems, we can use the reductions 
to obtain algorithms for the other two problems 
that require essentially the same time and space (the transformation 
only adds an additive overhead of $O(n^2/s)$ time).

\subsection{Obtaining Genuine Trade-Offs}
Even though the work of Asano~\etal~\cite{AsanoBuBuKoMuRoSc13}
gave a time-space trade-off for shortest paths in polygons,
it did so at the expense of a preprocessing step that 
triangulates the polygon, a task for which they could only
claim a $O(n^2)$ time algorithm, even if $s = \omega(1)$ cells
of workspace were available. Thus, at SoCG 2014, Tetsuo Asano asked whether
there is a more direct time-space trade-off for the shortest
path problem that does not go through such a preprocessing step.

Soon after, Har-Peled~\cite{HarPeled16} answered this question in the 
affirmative. He described a randomized algorithm that
uses the \emph{violator spaces} framework~\cite{GaertnerMaRuSk08} 
and a new polygon decomposition technique~\cite{HarPeled14} 
in order to compute the 
shortest path between any two points in $P$ in expected
time $O(n^2/s+n\log s\log^4(n/s))$, provided that $s$ cells of workspace
are available. 
Har-Peled's result, combined with the reductions of Asano~\etal and 
Aronov~\etal, also gives algorithms for the other two problems that
run in expected time $O(n^2/s+n\log s\log^5(n/s))$, using
$s$ cells of workspace~\cite{AronovKoPrvReRo16}. 

Very recently, Oh and Ahn~\cite{OhAh17} showed that
a similar result could also be obtained with a deterministic
algorithm. They explained how to find the trapezoidal decomposition of 
$P$ in $O(n^2/s)$ time, using $s$ cells of workspace. This 
decomposition, together with the method of 
Asano~\etal~\cite{AsanoBuBuKoMuRoSc13}, gives a $O(n^2/s)$ time
deterministic algorithm for obtaining a balanced partition of $P$ 
with $s$ chords using $s$ cells of workspace. 
Again, the reductions yield algorithms with similar running
times for both computing a triangulation and shortest paths in $P$,
see Table~\ref{tab:summary}.

\subsection{Special Cases}
Special classes of polygons have also been studied.
Asano and Kirkpatrick~\cite{AsanoKi13} showed that a 
\emph{monotone} polygon
can be triangulated in $O(n \log_s n)$ time, provided that
$s$ cells of workspace are available. This is a 
remarkable trade-off, since at the $s=O(1)$ extreme, it decreases 
the memory required by a linear fraction while only increasing the 
running time by a logarithmic factor. Moreover, it also has 
a gradual improvement of the running time as the
available space increases that matches the best bounds when 
$s=\Theta(n)$. The algorithm by Asano and Kirkpatrick
proceeds through a reduction to the \emph{all nearest larger
neighbors} (ANLN) problem: given a sequence $a_1, a_2, \dots, a_n$ 
of $n$ real numbers, find for each $a_i$, $i \in \{1, \dots, n\}$,
the closest indices $j < i$ and $k > i$ with $a_j, a_k > a_i$.

\begin{problem}
Asano and Kirkpatrick showed that their algorithm is
optimal for $s = O(1)$. Can the time-space trade-off be improved?
\end{problem}

\section{Other Problems in Simple Polygons}

In addition to shortest paths and triangulations, also problems
concerning visibility, convex hulls, and common tangents for
simple polygons have been studied in the limited workspace model.

\subsection{Visibility}
Another problem that has received heavy attention
in the limited workspace model is \emph{visibility} in simple 
polygons. Visibility problems have played a major role in 
computational geometry for a long time, see~\cite{Ghosh07} for an 
overview. 
The simplest problem in this family is as follows:
given a simple polygon $P$, a point $p\in P$, and an integer $k \geq 0$, 
compute 
the set of points in $P$ that are \emph{$k$-visible} from $p$, i.e., 
the set 
of points $q\in P$ for which the connecting segment $pq$ 
properly intersects the boundary of $P$ at most $k$ times. This 
set is called the \emph{$k$-visible region} of $p$ in $P$.
Thus, if we interpret the polygon as the walls of a building, 
the $0$-visible region of $p$ is the set of points that $p$ 
can see directly, without seeing through the walls. 
If we consider $2$-visibility, we allow the segment to leave 
(and re-enter) $P$ once, and so on. 

For $k=0$, the $0$-visibility problem can be solved in linear time 
and space, using a classic algorithm~\cite{JoeSi87}. For larger 
values of $k$, the problem can be solved in quadratic 
time~\cite{bajuelos2012hybrid}.\footnote{Note that the notion of 
$k$-visibility is slightly different in~\cite{bajuelos2012hybrid}:
they consider all $k$-visible points in the plane and not just the 
points inside the polygon.}
Several time-space trade-offs are known for computing the 
$k$-visible region. 
For $0$-visibility, Barba~\etal~\cite{BarbaKoLaSi14} presented an algorithm
that requires $O(1)$ cells of workspace and that runs in 
$O(n \bar r)$ time, where $\bar r$ is the number of reflex
vertices of $P$ that appear in the $0$-visible 
region of $p$.
Their work also contains a
time-space trade-off for finding
the $0$-visible region. 
More precisely, they describe 
a recursive method based on
their constant workspace algorithm that runs in
$O(nr/2^s + n \log^2 r)$ deterministic time and in $O(nr/2^s + n \log r)$
randomized expected time. Here, $r$ is the number of
reflex vertices in $P$. Their algorithm uses $s$ cells of workspace,
where $s$ is allowed to range from $1$ to $\log r$~\cite{BarbaKoLaSi14}. 
Only slightly later, a superset of the 
authors~\cite{BarbaKoLaSaSi15} gave an
improved algorithm for the $0$-visibility problem that runs
faster for specific combinations of $r$, $s$, and $n$. In
fact, Barba~\etal~\cite{BarbaKoLaSaSi15} 
discovered a much more general method for obtaining 
time-space trade-offs for a wide class of geometric
algorithms that they call \emph{stack-based} algorithms.
See below for a more detailed description of this method.

\begin{problem}
For $s = \log n$, the visibility region can be computed 
in $O(n\log^2 n)$ deterministic and $O(n \log n)$ expected time. 
For $s = n$, it takes $O(n)$ deterministic time. What happens for 
$s = \log n, \dots, n$?
Partial answers are known. For example, using the compressed stack 
framework mentioned below, we can achieve $O(n)$ time using 
$n^{\varepsilon}$ cells of workspace, for any 
fixed $\eps > 0$. What is the smallest value of $s$ for which we can compute the visibility region in linear time?
\end{problem}
For the general case of $k \in \{1, \dots, n\}$, 
Bahoo~\etal~\cite{BahooBaBoDuMu17}
provided a time-space trade-off that computes the $k$-visible region of  
$p$ in a polygonal region $P$ in $O\big((k+c)n/s + n\log s\big)$ time 
using $s$ cells of 
workspace, where $s$ may range from $1$ to $n$. Here, $1 \leq c \leq n$ 
is the number of ``critical'' vertices of
$P$, i.e., vertices, where the $k$-visible region may 
change.\footnote{The actual trade-off is a 
bit more complicated, and we chose a simpler bound 
for the sake of presentation.} This algorithm 
makes use of known time-space trade-offs for the
$k$-selection problem~\cite{ChanMuRa14} and requires a
careful analysis of the combinatorics of the $k$-visible 
region of $p$ in $P$.

Another notion of visibility was considered by 
Abrahamsen~\cite{Abrahamsen13}.  
Given a simple polygon $P$ and an edge $e$ of
$P$, the \emph{weak visibility region} of
$e$ in $P$ is the set of all the points 
inside $P$ that are visible from at least
one point on $e$.  Abrahamsen developed 
constant workspace algorithms
for edge-to-edge visibility. This leads to 
an $O(nm)$ time algorithm for finding the 
weak visibility region inside a given simple
polygon $P$ for an edge $e$ in constant
workspace, where $m$ denotes the size of the 
output. The result also gives an $O(n^2)$ time and constant
workspace algorithm for computing a minimum-link-path 
between two points in a simple polygon with $n$
vertices.

\subsection{The Compressed Stack Framework}
Barba~\etal~\cite{BarbaKoLaSaSi15}
provided a general method for \emph{stack-based}
algorithms in the limited workspace model.
Intuitively, a deterministic incremental algorithm is stack-based if its 
main data structure  takes the form of a stack. 
In addition to computing $0$-visibility regions,
classic examples from this category include the
algorithms for computing the convex hull of a simple
polygon by Lee~\cite{Lee83} or for triangulating a monotone 
polygon by Garey~\etal~\cite{GareyJoPrTa78}. Other applications in 
graphs of bounded treewidth are also known~\cite{banerjee2015timespace}.

The general trade-off is obtained by using a 
\emph{compressed stack} that explicitly stores
only a subset of the stack that is needed during the
computation and that recomputes the remaining
parts of the stack as they are needed.
Some delicate work goes into balancing the space
required for the partial stack and the time needed for 
reconstructing the other parts.
The upshot of applying this technique is as 
follows: given a stack-based algorithm that
on input size $n$ runs in $O(n)$ time and
uses a stack with $O(n)$ cells, one 
can obtain an 
algorithm that uses $s$ cells
of workspace and runs in $O(n^2\log n/2^s)$ 
time for $s = o(\log n)$ and 
in $n^{1+1/\log s}$ time for 
$s \geq \log n$.\footnote{Again, as above, the 
actual trade-off is more nuanced, 
but we simplified the bound to make it more digestible for the
casual reader. More details can be found in the original 
paper~\cite{BarbaKoLaSaSi15}.}
An experimental evaluation of the framework was conducted by
Baffier~\etal~\cite{BaffierDiKo18}.

\begin{problem}
Can we modify the compressed stack framework to 
compress other structures (e.g., queues, 
dequeues, or trees)? If so, what additional applications 
follow from these techniques?
\end{problem}

\subsection{Common Tangents}
The problem of finding the common tangents of two disjoint polygons can
also be solved with 
a limited amount of workspace.
Abrahamsen~\cite{Abrahamsen15} showed that one can find the
\emph{separating tangents} of two polygons 
with $n$ vertices and disjoint convex hulls  
in $O(n)$ time and $O(1)$ cells of workspace (if the convex hulls 
overlap, the algorithm
reports that a separating common tangent does not exist.)
His algorithm is stunningly simple, consisting essentially of a single
for-loop, but it requires a subtle 
analysis. 
In follow-up work, Abrahamsen and Walczack~\cite{AbrahamsenWa16} changed 
two lines in Abrahamsen's  algorithm 
to show that the problem of finding the \emph{outer
tangents} of two disjoint polygons with $n$ vertices can also be 
solved with
$O(n)$ time and $O(1)$ cells of workspace.
Combining these two results, one can decide in the same time and 
space complexity if the convex hulls of the given polygons are disjoint, 
overlapping, or nested.
\begin{problem}
Can we find the outer common tangents of two non-disjoint polygons in 
linear time using $O(1)$ words of workspace?
\end{problem}

\bibliographystyle{abbrv}
\bibliography{cws} 
\end{document}